\documentclass[11 pt]{article}
\usepackage{amsmath}
\usepackage{amsthm}
\usepackage{amssymb}
\usepackage{graphicx}
\usepackage{hyperref}
\usepackage{color}
\usepackage{float}
\usepackage{comment}

\usepackage{ulem}
\usepackage{enumerate}

\newtheorem{theorem}{Theorem}[section]
\newtheorem{definition}[theorem]{Definition}
\newtheorem{assumption}[theorem]{Assumption}
\newtheorem{lemma}[theorem]{Lemma}

\newtheorem{remark}[theorem]{Remark}
\numberwithin{equation}{section}

\newcommand{\R}{\mathbb{R}}

\newcommand{\Var}{\text{Var}}

\author{ 
	Shantanu Awasthi\footnote{Independent Scholar, Email: awasthi.shantanu@gmail.com}, 
	\quad Minglian Lin\footnote{Department of Mathematical Sciences, The University of Texas at El Paso, Email: mlin2@utep.edu}, 
	\quad Blair Faber\footnote{Independent Scholar, Email: Blairfaber2@gmail.com}, 
	\quad Michael Roberts\footnote{Independent Scholar, Email: mjroberts543@gmail.com},
	\quad Hassan Butt\footnote{Plaster School of Business, Missouri Southern State University, Email: Butt-H@mssu.edu}
	}
\date{}

\begin{document}
\title{\textsc{Enhancing the Black-Scholes Model for Option Valuation via Lévy Processes and Malliavin Calculus}}

\maketitle
\begin{abstract}
	The Black-Scholes model has been extensively used for option pricing, but exhibits limitations in its reliance on geometric Brownian motion and fixed volatility assumptions. This paper proposes an enhanced model incorporating stochastic volatility with jumps modeled by a Lévy process. Leveraging multidimensional Itô calculus, we derive a pricing formula for European call options under the new framework. Additionally, Malliavin calculus enables the derivation of an exact expression for at-the-money implied volatility. The proposed model is shown to better capture empirical features like volatility smiles. Analysis of VIX data demonstrates the model's ability to match observed market volatility. The integration of Lévy processes and Malliavin calculus represents a valuable advancement in addressing deficiencies in the classic Black-Scholes model. Further empirical testing is warranted to validate the approach across varying market conditions and option types.
\end{abstract}



\section{Introduction}
The Black-Scholes-Merton (BSM) model has long been employed by financial markets to price European options, despite its inherent imperfections and limitations. By utilizing inputs such as implied volatility, time, underlying asset price, strike price, and risk-free rate, the BSM model provides a pricing framework for both call and put options \cite{Humay}. It remains a prevalent choice for financial institutions and corporations in risk modeling and option valuation, particularly in the context of Employee Stock Options.

Option pricing has a significant influence on quantitative finance. In the BSM model,
the option value depends on future volatility of the stock rather than its expected return \cite{NV}.
One drawback of the classical BSM model is the mismatch between model-implied volatility and
market-observed implied volatility \cite{SenGupta}. 
The BSM model builds upon the foundation of Geometric Brownian Motion, which assumes continuous changes in asset prices over infinitesimal time intervals, where randomness governs the alteration of asset state. However, this model encounters key limitations, particularly in the usage of implied volatility. Historical volatility or implied volatility of similarly priced assets are employed to determine implied volatility, leading to potential inaccuracies and mispricing \cite{Carr}. Additionally, the model assumes constant volatility throughout the option's lifespan, disregarding the inherent variability of asset volatility during trading days. Notably, option prices often exhibit a volatility smile or skew, with higher implied volatility for out-of-the-money options and lower implied volatility for in-the-money options. This market phenomenon indicates an expectation of higher volatility for extreme price movements, which the BSM model fails to account for.

Furthermore, the BSM model relies on the assumption that stock price follows a log-normal distribution, which aligns with the expected returns of the stock distributed normally for an instant in time. However, this assumption overlooks the considerable variations in stock price distributions that deviate from normality, as previously discussed \cite{ShanIn}. Moreover, the model's applicability is limited to European options, unable to accommodate contracts with early exercise, such as American options. This limitation arises from the model's inability to determine the present value of a strike price when the time to expiration is uncertain.

This research paper aims to address the limitations associated with implied volatility in the BSM model by introducing a novel approach that incorporates a Lévy subordinator. By considering volatility as a stochastic process with nested jump terms, the proposed framework provides a more accurate depiction of market dynamics. Leveraging the principles of Itô's lemma for multidimensional Lévy processes and utilizing suitable approximations, we derive an expression for the price of a European call option. Additionally, this study explores the application of Malliavin calculus\cite{Applebaum}, a powerful mathematical tool, to obtain an exact expression for the implied volatility at-the-money (ATM) when the underlying asset adheres to the proposed model.

The paper's structure unfolds as follows:  In Section 2, we introduce the Lévy subordinator and demonstrate its incorporation into the option pricing framework, utilizing Itô's lemma and appropriate approximations. This section also  presents the main theorem, establishing the expression for the price of a European call option under the proposed model, accompanied by a thorough proof elucidating the derivations and approximations involved.
Section 3 focuses on the ATM implied volatility and investigates its exact expression within the proposed framework using Malliavin calculus. By exploiting the properties of Malliavin calculus, we provide an in-depth analysis of the relationship between implied volatility and the dynamics of the underlying asset.
Section 4 presents the data analysis using VIX data and strike price of 240, selection of strike price is arbitrary. Plot is this section shows the comparison between actual and computed implied volatility.



\section{Option Pricing}

Motivated by the empirical evidence that jumps in stock
returns co-occur with jumps in volatility, we introduce a two dimensional jump diffusion model in which
the log-price and a volatility proxy share common jump shocks.

\subsection{Notations}

\subsubsection{Noise and pre-jump evaluation}

Let $(\Omega,\mathcal F,(\mathcal F_t)_{t\in[0,T]},\mathbb Q)$ be a filtered probability space supporting
\begin{enumerate}[(i)]
\item a one-dimensional Brownian motion $(W_t)_{t\in[0,T]}$,
\item a Poisson random measure $N(dt,dz)$ on $[0,T]\times([-n,n]\setminus\{0\})$ with L\'evy measure $\nu(dz)$,
independent of $W$,
\end{enumerate}
and let $\tilde N(dt,dz):=N(dt,dz)-\nu(dz)\,dt$ be the compensated Poisson random measure.

\begin{definition}[Left limits, $s-$, and jumps]\label{def:cadlag_notation}
For a c\`adl\`ag process $X$, define the left limit (pre-jump value)
\[
X_{s-}:=\lim_{u\uparrow s}X_u,
\]
and the jump size
\[
\Delta X_s:=X_s-X_{s-}.
\]
The notation $s-$ means ``just before time $s$,'' i.e.\ evaluation at $X_{s-}$.
In particular, $X_s^{-}$ is synonymous with $X_{s-}$.
\end{definition}


\subsubsection{The coefficients $A_{11},A_{12},A_{21},A_{22}$}

Let $X_t=(X_t^1,X_t^2)$ be an $\R^2$ semimartingale whose continuous martingale part is driven by a Brownian motion
with diffusion vector (in our case) $\Sigma_t\in\R^{2\times 1}$. Then the continuous quadratic covariations satisfy
\begin{equation}\label{eq:cov_cont}
	d[X^i,X^j]^c_t = (\Sigma_t\Sigma_t^\top)_{ij}\,dt,\qquad i,j\in\{1,2\},
\end{equation}
and the second-order continuous term in It\^o's formula is
\begin{equation}\label{eq:Aij_general}
	\frac12\sum_{i,j=1}^2 (\Sigma_t\Sigma_t^\top)_{ij}\,\partial_{x_ix_j}f(t,X_{t-})\,dt
	=\sum_{i,j=1}^2 A_{ij}(t)\,\partial_{x_ix_j}f(t,X_{t-})\,dt,
\end{equation}
where
\begin{equation}\label{eq:Aij_def}
	A_{ij}(t):=\frac12(\Sigma_t\Sigma_t^\top)_{ij}.
\end{equation}
Thus $A_{ij}$ are generally time-dependent processes, not constants.

In our model,
since only $X^1$ has a Brownian term in \eqref{eq:X1_adapted} and $X^2$ has no Brownian term,
\[
\Sigma_t=\begin{pmatrix}\sigma_t\\ 0\end{pmatrix}
\quad\Rightarrow\quad
\Sigma_t\Sigma_t^\top=
\begin{pmatrix}
	\sigma_t^2 & 0\\
	0 & 0
\end{pmatrix},
\]
so
\begin{equation}\label{eq:Aij_model}
	A_{11}(t)=\tfrac12\sigma_t^2,\qquad A_{12}(t)=A_{21}(t)=0,\qquad A_{22}(t)=0.
\end{equation}


\subsubsection{It\^o formula for a two-dimensional L\'evy semimartingale}

Consider a two-dimensional semimartingale of the form
\begin{equation}\label{eq:general_2d}
	dX_t=b_t\,dt+\Sigma_t\,dW_t+\int_{-n}^{n}\gamma(t,X_{t-},z)\,\tilde N(dt,dz),
\end{equation}
where $b_t\in\R^2$ and $\gamma(\cdot)\in\R^2$.
Let $f\in C^{1,2}([0,T]\times\R^2)$. Then the It\^o-L\'evy formula (multidimensional) yields
\begin{equation}\label{eq:ito_2d}
	\begin{aligned}
		df(t,X_t)
		={}&\partial_t f(t,X_{t-})
		+\nabla f(t,X_{t-})^\top b_t
		+\tfrac12\mathrm{tr}\big((\Sigma_t\Sigma_t^\top)\nabla^2 f(t,X_{t-})\big)\\
		&
		+\int_{-n}^{n}\Big(f(t,X_{t-}+\gamma(t,X_{t-},z))-f(t,X_{t-})
		-\nabla f(t,X_{t-})^\top\gamma(t,X_{t-},z)\Big)\nu(dz)dt\\
		&+\nabla f(t,X_{t-})^\top\Sigma_t\,dW_t\\
		&+\int_{-n}^{n}\Big(f(t,X_{t-}+\gamma(t,X_{t-},z))-f(t,X_{t-})\Big)\tilde N(dt,dz).
	\end{aligned}
\end{equation}

\subsection{Adapted model specification under the risk-neutral measure}

\begin{assumption}\label{ass:bounded_jumps}
	The jump sizes are bounded: there exists $n>0$ such that all jumps are supported on $[-n,n]\setminus\{0\}$.
	Moreover, the jump integrands satisfy the standard square-integrability conditions with respect to
	$\nu(dz)\,dt$ so that all stochastic integrals below are well-defined.
\end{assumption}

Under the risk-neutral probability $\mathbb Q$, define the stock price by $S_t:=e^{X_t^1}$ where the
log-price $X^1$ follows the \textit{adapted} semimartingale SDE
\begin{equation}\label{eq:X1_adapted}
dX_t^1=\Big(r-\tfrac12\sigma_t^2-\kappa_t\Big)\,dt+\sigma_t\,dW_t+\int_{-n}^{n}\eta(t,z)\,\tilde N(dt,dz),
\qquad X_0^1\in\mathbb R,
\end{equation}
where $r$ is the constant interest rate, $\sigma_t$ is an adapted square-integrable process,
and $\eta(t,z)$ is the (log-)jump amplitude. The drift correction
\begin{equation}\label{eq:kappa_def}
\kappa_t := \int_{-n}^{n}\Big(e^{\eta(t,z)}-1-\eta(t,z)\Big)\,\nu(dz)
\end{equation}
ensures that the discounted price $e^{-rt}S_t$ is a $\mathbb Q$-martingale (see Lemma~\ref{lem:mart} below).
This resolves the missing/unclear compensator adjustment in the jump-diffusion drift.

\begin{lemma}[Risk-neutral martingale condition]\label{lem:mart}
Assume 
$$\int_{-n}^{n}\big(e^{\eta(t,z)}-1-\eta(t,z)\big)\nu(dz)<\infty, \quad \forall t\in[0,T],$$
then $e^{-rt}S_t$ is a $\mathbb Q$-martingale.
\end{lemma}

\begin{proof}
Applying the exponential It\^o formula for jump semimartingales to $S_t=e^{X_t^1}$ with dynamics
\eqref{eq:X1_adapted} yields
\[
\frac{dS_t}{S_{t-}} = r\,dt + \sigma_t\,dW_t + \int_{-n}^{n}\big(e^{\eta(t,z)}-1\big)\,\tilde N(dt,dz),
\]
because the compensator term $\kappa_t$ in \eqref{eq:kappa_def} cancels the drift contribution of jumps.
Therefore, $d(e^{-rt}S_t)$ has zero drift and is a local martingale; integrability yields a martingale.
\end{proof}

\paragraph{Volatility proxy with common jumps}
Let $Y_t$ be a positive adapted variance factor (or variance proxy) driven by the same jump measure $N$:
\begin{equation}\label{eq:Y_adapted}
dY_t = a(b-Y_t)\,dt + \int_{-n}^{n}\xi(t,z)\,N(dt,dz),\qquad Y_0>0,
\end{equation}
where $a>0$, $b>0$ and $\xi(t,z)\ge 0$ is the jump size mapping for $Y$. This specification induces
co-jumps in return and volatility through the common $N$.

We define the (scaled) volatility proxy $X_t^2$ in the same functional form,
but now with adapted $Y_t$:
\begin{equation}\label{eq:X2_def_paper}
X_t^2 := \sqrt{\frac{Y_t}{2n\,(T+\delta-t)}},\qquad \delta>0.
\end{equation}
Since $Y_t$ is adapted and $t\mapsto (T+\delta-t)^{-1}$ is deterministic, $X^2$ is adapted.


\subsection{Pricing representation}

Let $g(x_1)=(e^{x_1}-K)^+$ be the European call payoff and define the option price
\[
V_t:=\mathbb E_t^{\mathbb Q}\Big[e^{-r(T-t)}g(X_T^1)\Big].
\]
Let $u(t,x_1,x_2)$ denote the Black-Scholes call price evaluated at state $(x_1,x_2)$:
\begin{equation}\label{eq:u_def}
u(t,x_1,x_2):=\mathrm{BS}(t,T,e^{x_1},K,r,x_2),
\qquad u(T,x_1,x_2)=g(x_1).
\end{equation}          

\begin{theorem}\label{thm:2.2}
Assume the adapted model \eqref{eq:X1_adapted} - \eqref{eq:Y_adapted} and define $X_t=(X_t^1,X_t^2)$ by
\eqref{eq:X2_def_paper}. Then
\begin{equation}\label{eq:repr_main}
V_t
=
u(t,X_t)
+\mathbb E_t^{\mathbb Q}\!\left[\int_t^T e^{-r(s-t)}\,(\mathcal C_s u)\big(s,X_{s-}\big)\,ds\right],
\end{equation}
where the correction operator $\mathcal C_s$ equals
\begin{align}\label{eq:C_operator}
(\mathcal C_s u)(s,x)
=
& \underbrace{\frac12\Big(\sigma_s^2-x_2^2\Big)\Big(\partial_{x_1x_1}u-\partial_{x_1}u\Big)(s,x)}_{\text{variance mismatch}}
\;+\;
\underbrace{\frac{x_2}{2(T+\delta-s)}\,\partial_{x_2}u(s,x)}_{\text{time-scaling drift}} \nonumber \\
& \;+\;
\underbrace{\int_{-n}^{n}\mathcal J_u(s,x,z)\,\nu(dz)}_{\text{jump compensator}},
\end{align}
with jump compensator integrand
\begin{equation}\label{eq:J_def}
\mathcal J_u(s,x,z)
:=
u\big(s,x+\gamma(s,x,z)\big)-u(s,x)-\nabla_x u(s,x)\cdot\gamma(s,x,z),
\end{equation}
and jump amplitude vector (for $X=(X^1,X^2)$)
\begin{equation}\label{eq:gamma_model}
\gamma(s,x,z)
:=
\left(
\eta(s,z),\;
\sqrt{\frac{x_2^2\,2n(T+\delta-s)+\xi(s,z)}{2n(T+\delta-s)}}-x_2
\right).
\end{equation}
\end{theorem}

\begin{proof}
We present the proof in explicit steps, using the two-dimensional It\^o-L\'evy formula \eqref{eq:ito_2d}.Define
\[
F(s,x):=e^{-rs}u(s,x),\qquad x=(x_1,x_2)\in\R^2.
\]
Then
\begin{equation}\label{eq:F_derivs}
\partial_sF=e^{-rs}(\partial_su-ru),\qquad
\nabla_xF=e^{-rs}\nabla_xu,\qquad
\nabla_x^2F=e^{-rs}\nabla_x^2u.
\end{equation}
From \eqref{eq:X1_adapted} and \eqref{eq:X2_def_paper}, $X$ is c\`adl\`ag with Brownian term only in
the first component. Hence its continuous diffusion vector is
\[
\Sigma_s=\begin{pmatrix}\sigma_s\\0\end{pmatrix},
\quad\text{so that}\quad
A_{11}(s)=\tfrac12\sigma_s^2,\ A_{12}(s)=A_{21}(s)=A_{22}(s)=0
\]
as in \eqref{eq:Aij_model}. The jump of $X^1$ at time $s$ is $\Delta X_s^1=\eta(s,z)$ when $N$ jumps by $z$.
The jump of $Y$ is $\Delta Y_s=\xi(s,z)$, therefore the jump in $X^2=\sqrt{Y/(2n(T+\delta-s))}$ is exactly
the second component in \eqref{eq:gamma_model}. Thus the jump amplitude vector for $X$ is $\gamma(s,X_{s-},z)$.
Applying \eqref{eq:ito_2d} with $f=F$ gives
\[
d\big(F(s,X_s)\big)=\text{(drift)}\,ds+\text{(local martingale)},
\]
where the drift term is
\begin{align}
\label{eq:drift_term}
e^{-rs}\Big[
(\partial_su-ru)(s,X_{s-})
&+\nabla_xu(s,X_{s-})^\top b_s
+\tfrac12\sigma_s^2\,\partial_{x_1x_1}u(s,X_{s-}) \nonumber\\
&+\int_{-n}^{n}\Big(u(s,X_{s-}+\gamma)-u(s,X_{s-})-\nabla_xu(s,X_{s-})^\top\gamma\Big)\nu(dz)
\Big]ds
\end{align}
and the local martingale term is the sum
\begin{align}
\label{eq:lm_term}
e^{-rs}\,\partial_{x_1}u(s,X_{s-})\,\sigma_s\,dW_s
+\;e^{-rs}\!\int_{-n}^{n}\Big(u(s,X_{s-}+\gamma)-u(s,X_{s-})\Big)\tilde N(ds,dz).
\end{align}
Integrating \eqref{eq:drift_term} - \eqref{eq:lm_term} over $[t,T]$ yields
\begin{equation}\label{eq:integrate}
e^{-rT}u(T,X_T)-e^{-rt}u(t,X_t)
=\int_t^T \text{(drift)}\,ds\;+\;\int_t^T \text{(local martingale)}.
\end{equation}
Taking $\mathbb E_t^{\mathbb Q}[\cdot]$ of \eqref{eq:integrate} and using that the conditional expectation
of the local martingale term is $0$ (by standard localization/integrability), we obtain
\begin{equation}\label{eq:after_condexp}
\mathbb E_t^{\mathbb Q}\!\left[e^{-rT}u(T,X_T)\right]
=
e^{-rt}u(t,X_t)
+\mathbb E_t^{\mathbb Q}\!\left[\int_t^T \text{(drift)}\,ds\right].
\end{equation}
Since $u(T,x_1,x_2)=g(x_1)$, multiplying both sides of \eqref{eq:after_condexp} by $e^{rt}$ gives
\[
V_t
=
u(t,X_t)
+\mathbb E_t^{\mathbb Q}\!\left[\int_t^T e^{-r(s-t)}\,\Big((\partial_su+\mathcal L_su-ru)(s,X_{s-})\Big)\,ds\right],
\]
where $\mathcal L_s$ is the generator induced by the dynamics of $X$.
Finally, because $u$ is the Black-Scholes price with volatility input $x_2$,
the Black-Scholes PDE implies that the $(t,x_1)$-terms cancel when $\sigma_s^2$ is replaced by $x_2^2$,
leaving the variance mismatch term
$\frac12(\sigma_s^2-x_2^2)(\partial_{x_1x_1}u-\partial_{x_1}u)$.
The deterministic time scaling in $x_2=\sqrt{Y/(2n(T+\delta-s))}$ contributes the drift term
$\frac{x_2}{2(T+\delta-s)}\partial_{x_2}u$,
and the jump compensator contributes $\int \mathcal J_u\,\nu(dz)$ as in \eqref{eq:J_def}.
This yields \eqref{eq:repr_main} - \eqref{eq:gamma_model}.
\end{proof}

\begin{remark}[What each symbol means in the theorem/proof]
\begin{enumerate}[(i)]
\item $X_{s-}$ is the pre-jump value (left limit), and $s-$ means evaluation at $X_{s-}$.
\item $\Delta X_s=X_s-X_{s-}$ is the jump size.
\item $A_{ij}(s)=\tfrac12(\Sigma_s\Sigma_s^\top)_{ij}$; in this model $A_{11}=\tfrac12\sigma_s^2$ and
$A_{12}=A_{21}=A_{22}=0$.
\end{enumerate}
\end{remark}

\section{Implied Volatility}

In this section, we study at-the-money implied volatility of a European call when stock price is driven by both a Brownian motion and a L\'evy process. 

For simplicity, we assume that the interest rate is zero.
At time $t$, we consider the log-price $X$ of a stock  under a risk neutral probability measure $\mathbb{Q}$:
\begin{align}
\label{se11}
X_t = X_0 -\frac{1}{2} \int_{0}^{t} \sigma_s^2 \,ds + \int_{0}^{t} \sigma_s \,dW_s + \int_{0}^{t} \int_{\mathbb{R}_0} \sigma^2(s,z) \tilde{N}(ds,dz),\quad t \in[0,T],                           
\end{align}
where $T$ is terminal time, $X_0$ is current log-price, 
$\sigma_s$ is a square integrable and right continuous stochastic process adapted to the filtration generated by standard Brownian motion $W$ defined on a complete probability space $(\Omega,\mathcal F,(\mathcal F_t)_{t\in[0,T]},\mathbb Q)$, 
$\mathbb{R}_0 = \mathbb{R}\setminus \{0\}$,
$z$ is generic jump size, 
$\sigma^2(s,z)$ is a Skorohod integrable stochastic process, 
and $\tilde{N}$ is a compensated Poisson process.
We represent the future average volatility $v$ by
\begin{align*}
	v_t = \sqrt{\frac{1}{T-t}\int_{t}^{T} \int_{\mathbb{R}_0} \sigma^{2} (s,z) \tilde{N}(ds,dz)}.
\end{align*}
In the case of zero interest rate, the price of an European call option $V_t$ with strike price $K$ is
\begin{align*}
    V_t = E_t\big[(e^{X_T}-K)^{+}\big],
\end{align*}
where $E_t[\cdot] = E[\cdot|\mathcal{F}_{t}]$ with respect to risk neutral probability $\mathbb{Q}$. 
Under the  Black-Scholes model, the well known formula of the price of European call option is 
\begin{align*}
	BS(t,T,x,k,\sigma) = e^x \Phi(d_+(k,\sigma)) - e^k \Phi(d_-(k,\sigma)),    
\end{align*}
where $x = X_0$, $k=\ln K$, $\sigma$ is constant volatility, $\Phi$ is the cumulative distribution function of standard normal distribution $N(0,1)$, and 
\begin{align*}
	d_\pm(k,\sigma) = \frac{k_t^*-k}{\sigma\sqrt{T-t}}\pm \frac{\sigma\sqrt{T-t}}{2},
\end{align*}
where $k_t^*$ is the at-the-money strike. If the interest rate is zero, $ k_t^* = x$.
Referring to \cite{Eliya}, we define the inverse function $BS^{-1}(t,T,x,k,\cdot)$ as
\begin{align} \label{Inverse}
	BS\big(t,T,x,k,BS^{-1}(t,T,x,k,\lambda)\big) = \lambda, \quad \lambda > 0.
\end{align}
For fixed $t$, $T$, $X_t$, and $k$, we define the implied volatility $I(t,T,X_t,k)$ as
\begin{align}\label{Implied}
	BS\big(t,T,X_t,k,I(t,T,X_t,k)\big) = V_t.
\end{align}
Comparing \eqref{Inverse} with \eqref{Implied},
\begin{align}\label{I}
	I(t,T,X_t,k) =  BS^{-1}(t,T,X_t,k,V_t).
\end{align}
In the follow-up, we let $\mathbb{D}^{1,2}$ be the domain of Malliavin derivative operator $D$.
As $\mathbb{D}^{1,2}$ is a dense subset of $L^2(\Omega)$ and $D$ is a closed and unbounded operator, we denote $\mathbb{L}^{1,2} = L^{2}([0,T]; \mathbb{D}^{1,2})$.

\begin{theorem}\label{se12}
If the following hypotheses hold,
\begin{enumerate}[(1)]
	\item There exist positive constants $a$ and $b$ such that $a\leq \sigma_t \leq b$, $\forall t \in[0,T]$.
	\item $\sigma^2 \in \mathbb{L}^{1,2} $.
\end{enumerate}
then, in the case of uncorrelated stocks, the model \eqref{se11} results the following at-the-money implied volatility
\begin{align*}
	& I(t,T,X_t,k_t^*) = \\
	& E_t[v_t] \\
	& - E_t \Bigg[
	\frac{1}{32(T-t)}\int_{t}^{T} 
	\frac{BS^{-1}(k_t^{*},\Lambda_s)}{\big(\Phi^{'}\big(d_+(k_t^{*}, BS^{-1}(k_t^{*},\Lambda_s))\big)\big)^2}
	\bigg(E_s\bigg[\frac{\Phi^{'}(d_+(k_t^{*},v_t))}{v_t} \int_{s}^{T} D_s  \sigma_r^2 dr\bigg]\bigg)^2 ds \\
	& \qquad \
	+ \int_{t}^{T} \int_{\mathbb{R}_0} \bigg\{ BS^{-1}(k_t^{*}, \Lambda_s + U_{s,z}) - BS^{-1}(k_t^{*}, \Lambda_s) \\
	& \qquad \
	- \frac{1}{2(T-t)\Phi^{'}\big(d_+(k_t^{*}, BS^{-1}(k_t^{*},\Lambda_s))\big)} 
	E_s\bigg[\frac{\Phi^{'}(d_+(k_t^{*},v_t))}{v_t} \int_{s}^{T} D_{s,z}  \sigma^2(r,z) dr\bigg] \bigg\} \nu(dz)ds \\
	& \qquad \
	+ \int_{t}^{T} \int_{\mathbb{R}_0} \bigg\{ BS^{-1}(k_t^{*}, \Lambda_s + U_{s,z}) - BS^{-1}(k_t^{*}, \Lambda_s) \bigg\} \tilde{N}(ds,dz)
	\Bigg],
\end{align*}
where
\begin{align*}
	\Lambda_s &= E_s\big[BS(t,T,X_t,k_t^{*},v_t)\big],\\
	U_{s,z}
	&= E_s\bigg[e^{X_t} \frac{\Phi^{'}(d_+(k_t^{*},v_t))}{2v_t\sqrt{T-t}} \int_{s}^{T} D_{s,z}  \sigma^2(r,z) dr\bigg].
\end{align*}

\end{theorem}
\begin{proof}	
	We prove this theorem with a similar argument as in \cite{Eliya}. In the case of uncorrelated stocks, the Hull-White formula gives that
	\begin{align} \label{H_W}
		V_t = E_t[BS(t,T,X_t,k_t^*,v_t)].
	\end{align}
	We denote $BS^{-1}(k_t^*,V_t) = BS^{-1}(t,T,X_t,k_t^*,V_t)$. Then
\begin{align}
	I(t,T,X_t,k_t^*)
	=\ & BS^{-1}(k_t^*,V_t), \qquad \text{by \eqref{I}} \nonumber\\
	=\ & E_t\big[BS^{-1}\big(k_t^{*},E_t[BS(t,T,X_t,k_t^*,v_t)]\big) \big], \qquad \text{by \eqref{H_W}} \nonumber\\
	=\ & E_t\Big[BS^{-1}\big(k_t^{*},E_t[BS(t,T,X_t,k_t^*,v_t)]\big) 
	- BS^{-1}\big(k_t^{*}, BS(t,T,X_t,k_t^*,v_t)\big) \nonumber\\
	&\ \quad + BS^{-1}\big(k_t^{*}, BS(t,T,X_t,k_t^*,v_t)\big) \Big] \nonumber\\
	=\ & E_t \Big[BS^{-1}\big(k_t^{*},E_t[BS(t,T,X_t,k_t^*,v_t)]\big) 
	- BS^{-1}\big(k_t^{*}, BS(t,T,X_t,k_t^*,v_t)\big) \Big] \nonumber\\
	& + E_t[v_t].	\label{II}
\end{align}
By the Clark-Ocone formula for combined Gaussian and pure jump L\'evy noise (Theorem 12.20 in \cite{Mallcalc}), we have
\begin{align*}
    BS(t,T,X_t,k_t^{*},v_t) = E_t[BS(t,T,X_t,k_t^{*},v_t)] + \int_{t}^{T} U_s dW_s + \int_{t}^{T} \int_{\mathbb{R}_0} U_{s,z} \tilde{N}(ds,dz),
\end{align*}
where
\begin{align}
	& U_s 
	= E_s[D_s BS(t,T,X_t,k_t^*,v_t)]
	= E_s\bigg[e^{X_t} \frac{\Phi^{'}(d_+(k_t^{*},v_t))}{2v_t\sqrt{T-t}} \int_{s}^{T} D_s  \sigma_r^2 dr\bigg],	\label{Us}\\
	& U_{s,z}
	= E_s[D_{s,z} BS(t,T,X_t,k_t^*,v_t)]
	= E_s\bigg[e^{X_t} \frac{\Phi^{'}(d_+(k_t^{*},v_t))}{2v_t\sqrt{T-t}} \int_{s}^{T} D_{s,z}  \sigma^2(r,z) dr\bigg].	\label{Usz}
\end{align}
Denote $\Lambda_s = E_s[BS(t,T,X_t,k_t^{*},v_t)]$. 
By the It\^o formula for the It\^o-L\'evy process (Theorem 9.4 in \cite{Mallcalc}), we obtain
\begin{align*}
	& E_t \Big[BS^{-1}\big(k_t^{*},E_t[BS(t,T,X_t,k_t^*,v_t)]\big) - BS^{-1}\big(k_t^{*}, BS(t,T,X_t,k_t^*,v_t)\big) \Big]\\
	=\ & E_t\big[BS^{-1}(k_t^{*},\Lambda_t) - BS^{-1}(k_t^{*},\Lambda_T)\big] \\
	=\ & - E_t\big[BS^{-1}(k_t^{*},\Lambda_T) - BS^{-1}(k_t^{*},\Lambda_t)\big] \\
	=\ & - E_t \bigg[\int_{t}^{T} \frac{\partial BS^{-1}(k_t^{*}, \Lambda_s)}{\partial \Lambda_s} U_s dW_s + \frac{1}{2}\int_{t}^{T} \frac{\partial^2 BS^{-1}(k_t^{*}, \Lambda_s)}{\partial \Lambda_s^2} U_s^{2} ds \\
	& \qquad \quad
	+ \int_{t}^{T} \int_{\mathbb{R}_0} \Big\{ BS^{-1}(k_t^{*}, \Lambda_s + U_{s,z}) - BS^{-1}(k_t^{*}, \Lambda_s) - \frac{\partial BS^{-1}(k_t^{*}, \Lambda_s)}{\partial \Lambda_s}U_{s,z} \Big\} \nu(dz)ds \\
	& \qquad \quad
	+ \int_{t}^{T} \int_{\mathbb{R}_0} \Big\{ BS^{-1}(k_t^{*}, \Lambda_s + U_{s,z}) - BS^{-1}(k_t^{*}, \Lambda_s) \Big\} \tilde{N}(ds,dz)
	\bigg],
\end{align*}
where $\nu$ is the L\'evy measure. 
The proof of Proposition 3.1 in \cite{Eliya} had already given that
\begin{align*}
	E_t\bigg[\int_{t}^{T} \frac{\partial BS^{-1}(k_t^{*}, \Lambda_s)}{\partial \Lambda_s} U_s dW_s\bigg] = 0.
\end{align*}
Hence,
\begin{align}
	& E_t \Big[BS^{-1}\big(k_t^{*},E_t[BS(t,T,X_t,k_t^*,v_t)]\big) - BS^{-1}\big(k_t^{*}, BS(t,T,X_t,k_t^*,v_t)\big) \Big] \nonumber\\
	=\ & - E_t \bigg[\frac{1}{2}\int_{t}^{T} \frac{\partial^2 BS^{-1}(k_t^{*}, \Lambda_s)}{\partial \Lambda_s^2} U_s^{2} ds \nonumber\\
	& \qquad \quad
	+ \int_{t}^{T} \int_{\mathbb{R}_0} \Big\{ BS^{-1}(k_t^{*}, \Lambda_s + U_{s,z}) - BS^{-1}(k_t^{*}, \Lambda_s) - \frac{\partial BS^{-1}(k_t^{*}, \Lambda_s)}{\partial \Lambda_s}U_{s,z} \Big\} \nu(dz)ds \nonumber\\
	& \qquad \quad
	+ \int_{t}^{T} \int_{\mathbb{R}_0} \Big\{ BS^{-1}(k_t^{*}, \Lambda_s + U_{s,z}) - BS^{-1}(k_t^{*}, \Lambda_s) \Big\} \tilde{N}(ds,dz)
	\bigg]. \label{Et}
\end{align}
In addition, we refer the $8$th line and the equation 3.2 in \cite{Eliya} to get the followings:
\begin{align}
	\frac{\partial BS^{-1}(k_t^{*}, \Lambda_s)}{\partial \Lambda_s}
	& = \frac{1}{e^{X_t}\Phi^{'}\big(d_+(k_t^{*}, BS^{-1}(k_t^{*},\Lambda_s))\big)\sqrt{T-t}},	\label{BS-1}\\
	\frac{\partial^2 BS^{-1}(k_t^{*}, \Lambda_s)}{\partial \Lambda_s^2}
	& = \frac{BS^{-1}(k_t^{*},\Lambda_s)}{4e^{2X_t}\big(\Phi^{'}\big(d_+(k_t^{*}, BS^{-1}(k_t^{*},\Lambda_s))\big)\big)^2}. \label{BS-2}
\end{align}
Finally, we combine the equations \eqref{II}, \eqref{Us}, \eqref{Usz}, \eqref{Et}, \eqref{BS-1}, and \eqref{BS-2} to to complete the proof.
\end{proof}

We now estimate at the money implied volatility under specific framework. We assume that the volatility follows a mean-reverting Ornstein-Uhlenbeck process, i.e. Stein-Stein model.
Under Stein-Stein Model, volatility process assumes following form
\begin{equation}
\label{SS}
    d\sigma_{t} = -\alpha(m-\sigma_{t}) dt + cdW_{t},
\end{equation}
where $\alpha$, $m$, and $c$ are positive real constants and $W$ is a standard Brownian motion.
Here we assume that, $W_{t} = \rho B_{t} + \sqrt{1-\rho^{2}}Z_{t}$, for $\rho \in (-1,1)$ and for two independent standard Brownian motions $Z$  and $B$.
The analytical solution of equation \eqref{SS} is given by
\begin{equation}
    \label{SSS}
    \sigma_{s} = m + (\sigma_{t} - m) e^{-\alpha(s-t) }+ c \int_{t}^{s} e^{-\alpha(s-u)}dW_u.
\end{equation}
In the following theorem, we use Malliavin calculus to derive an exact expression of the at-the-money implied volatility.

\begin{theorem}
\label{IV}
Assuming the volatility process $ \sigma $ in the model \eqref{se11} follows Stein-Stein model, the at-the-money implied volatility $I(t,T,X_t,k_t^{*})$ is given by
\begin{align*}
	& I(t,T,X_t,k_t^*) = \\
	& E_t[v_t] \\
	& - E_t \Bigg[
	\frac{1}{32(T-t)}\int_{t}^{T} 
	\frac{BS^{-1}(k_t^{*},\Lambda_s)}{\big(\Phi^{'}\big(d_+(k_t^{*}, BS^{-1}(k_t^{*},\Lambda_s))\big)\big)^2}
	\bigg(E_s\bigg[\frac{\Phi^{'}(d_+(k_t^{*},v_t))}{v_t} \int_{s}^{T} D_s  \sigma_r^2 dr\bigg]\bigg)^2 ds \\
	& \qquad \
	+ \int_{t}^{T} \int_{\mathbb{R}_0} \bigg\{ BS^{-1}(k_t^{*}, \Lambda_s + U_{s,z}) - BS^{-1}(k_t^{*}, \Lambda_s) \\
	& \qquad \quad
	- \frac{1}{2(T-t)\Phi^{'}\big(d_+(k_t^{*}, BS^{-1}(k_t^{*},\Lambda_s))\big)} 
	E_s\bigg[\frac{\Phi^{'}(d_+(k_t^{*},v_t))}{v_t} \int_{s}^{T} D_{s,z}  \sigma^2(r,z) dr\bigg] \bigg\} \nu(dz)ds \\
	& \qquad \
	+ \int_{t}^{T} \int_{\mathbb{R}_0} \bigg\{ BS^{-1}(k_t^{*}, \Lambda_s + U_{s,z}) - BS^{-1}(k_t^{*}, \Lambda_s) \bigg\} \tilde{N}(ds,dz)
	\Bigg],
\end{align*}
where
\begin{align*}
	\Lambda_s &= E_s\big[BS(t,T,X_t,k_t^{*},v_t)\big],\\
	U_{s,z}
	&= E_s\bigg[e^{X_t} \frac{\Phi^{'}(d_+(k_t^{*},v_t))}{2v_t\sqrt{T-t}} \int_{s}^{T} D_{s,z}  \sigma^2(r,z) dr\bigg].
\end{align*}
\end{theorem}

\begin{proof}
We have already shown the explicit expression for implied volatility $I(t,T,X_t,k_t^{*})$ in Theorem \ref{se12}. Now using equation \eqref{SSS} for volatility process, we compute Malliavin derivative as
\begin{align*}
    D_s \sigma_r^{2}
    & = 2\sigma_r\ D_s \sigma_r \\
    & = 2\sigma_r\ D_s \bigg( m + (\sigma_{t} - m) e^{-\alpha(r-t) }+ c \int_{t}^{r} e^{-\alpha(r-u)} \Big(\rho dB_{u} + \sqrt{1-\rho^{2}}dZ_{u} \Big) \bigg) \\
    & = 2\sigma_r  c e^{-\alpha(r-s)} \Big(\rho + \sqrt{1-\rho^{2}}\Big).
\end{align*}
Then 
\begin{align*}
	\int_{s}^{T}D_s \sigma_r^{2}dr = 2 \sigma_r c\rho \int_{s}^{T} e^{-\alpha(r-s)}dr = \frac{2\sigma_r c\rho  (e^{-\alpha(T-s)}-1)}{- \alpha}.
\end{align*}

\end{proof}

\section{Data Analysis}
We use the Cboe VIX index, from January 3, 2022 to September 14, 2022, as a market-based proxy for near-term ATM implied volatility. 
	The statistics of the data are summarized in Table \ref{tab:vix_summary_sigma}.
	All calculations are performed by Python 3.14.0.
\begin{table}[H]
	\centering
	\caption{Summary statistics of VIX and volatility proxy $\sigma_t=\mathrm{VIX}_t/100$ (Jan 3, 2022 - Sep 14, 2022).}
	\begin{tabular}{lrrrrr}
		\hline
		Statistic & Mean & Median & Std.\ Dev. & Min & Max \\
		\hline
		VIX (index) & 25.63 & 25.58 & 4.26 & 16.60 & 36.45 \\
		$\sigma_t=\mathrm{VIX}_t/100$ & 0.2563 & 0.2558 & 0.0426 & 0.1660 & 0.3645 \\
		\hline
	\end{tabular}
	\label{tab:vix_summary_sigma}
\end{table}

\paragraph{Stein-Stein (OU) dynamics on volatility.}

On each trading day $t_i$, we define the observed volatility level
\begin{equation}
\sigma_{t_i}^{\mathrm{mkt}} := \frac{\mathrm{VIX}_{t_i}}{100},
\end{equation}
so that $\sigma_t$ denotes volatility; the instantaneous variance is $v_t=\sigma_t^2$.
We model $\sigma_t$ as an Ornstein-Uhlenbeck process
\begin{equation}\label{eq:ou_sigma}
d\sigma_t=\alpha(m-\sigma_t)\,dt+c\,dW_t,\qquad \alpha>0,
\end{equation}
and estimate $(\alpha,m,c)$ from the exact discrete-time transition at step $\Delta$ (daily data: $\Delta=1/252$ years):
\begin{equation}\label{eq:ou_exact}
\sigma_{t_{i+1}} = m + (\sigma_{t_i}-m)e^{-\alpha\Delta} + \varepsilon_i,\qquad
\varepsilon_i\sim \mathcal N\!\left(0,\;\frac{c^2}{2\alpha}\bigl(1-e^{-2\alpha\Delta}\bigr)\right).
\end{equation}
Equivalently, with $b=e^{-\alpha\Delta}$ and $a=m(1-b)$, \eqref{eq:ou_exact} is the AR(1)
\begin{equation}\label{eq:ar1_sigma}
\sigma_{t_{i+1}}=a+b\,\sigma_{t_i}+\varepsilon_i,\qquad 
\Var(\varepsilon_i)=q=\frac{c^2}{2\alpha}(1-b^2).
\end{equation}
We estimate $(a,b)$ by OLS in \eqref{eq:ar1_sigma} (equivalently, Gaussian MLE under \eqref{eq:ou_exact})
and map back to continuous-time parameters via
\begin{equation}\label{eq:mapback_sigma}
\hat\alpha=-\frac{\ln \hat b}{\Delta},\qquad
\hat m=\frac{\hat a}{1-\hat b},\qquad
\hat c=\sqrt{\frac{2\hat\alpha\,\hat q}{1-\hat b^2}},
\end{equation}
where $\hat q$ is the mean squared residual from \eqref{eq:ar1_sigma}.

\begin{table}[H]
\centering
\caption{OU (Stein-Stein) calibration on volatility $\sigma_t$ using the exact discrete-time transition with $\Delta=1/252$.}
\begin{tabular}{lrr}
\hline
Parameter & Estimate & Notes \\
\hline
$\hat\alpha$ & 31.0046 & Mean-reversion speed (per year) \\
$\hat m$ & 0.2609 & Long-run mean of volatility \\
$\hat c$ & 0.3153 & Vol-of-vol (diffusion coefficient) \\
Half-life & 5.6338 & Trading days ($\ln 2/\hat\alpha$ in years) \\
\hline
\end{tabular}
\label{tab:ou_params_sigma}
\end{table}

Table 2 reports the estimated parameters of the Stein-Stein stochastic volatility model calibrated to the VIX-derived volatility series over the period January 3, 2022 to September 14, 2022. The estimated long-run mean volatility is $\hat{m}=0.2609$, which is close to the sample average of the observed volatility proxy reported in Table 1. This indicates that the Ornstein-Uhlenbeck specification is capable of reproducing the average volatility level observed during the sample period.
The estimated mean-reversion parameter $\hat{\alpha}=31.0046$ implies rapid adjustment of volatility toward its long-run equilibrium. In particular, the estimated half-life of volatility shocks is approximately $5.6338$ trading days, suggesting that periods of unusually high or low volatility tend to dissipate quickly. This finding is consistent with the empirical literature documenting the mean-reverting behavior of equity market volatility.
The estimated volatility-of-volatility parameter $\hat{c}=0.3153$ is relatively large, indicating substantial fluctuations in the volatility process. Such behavior is not surprising during the sample period, which coincides with heightened uncertainty caused by persistent inflation concerns, aggressive monetary tightening by the Federal Reserve, and geopolitical tensions associated with the Russia-Ukraine conflict. These events contributed to frequent changes in market expectations and increased volatility clustering.

\begin{figure}[H]
	\begin{minipage}{\textwidth}
	\centering
	\includegraphics[width=0.75\linewidth]{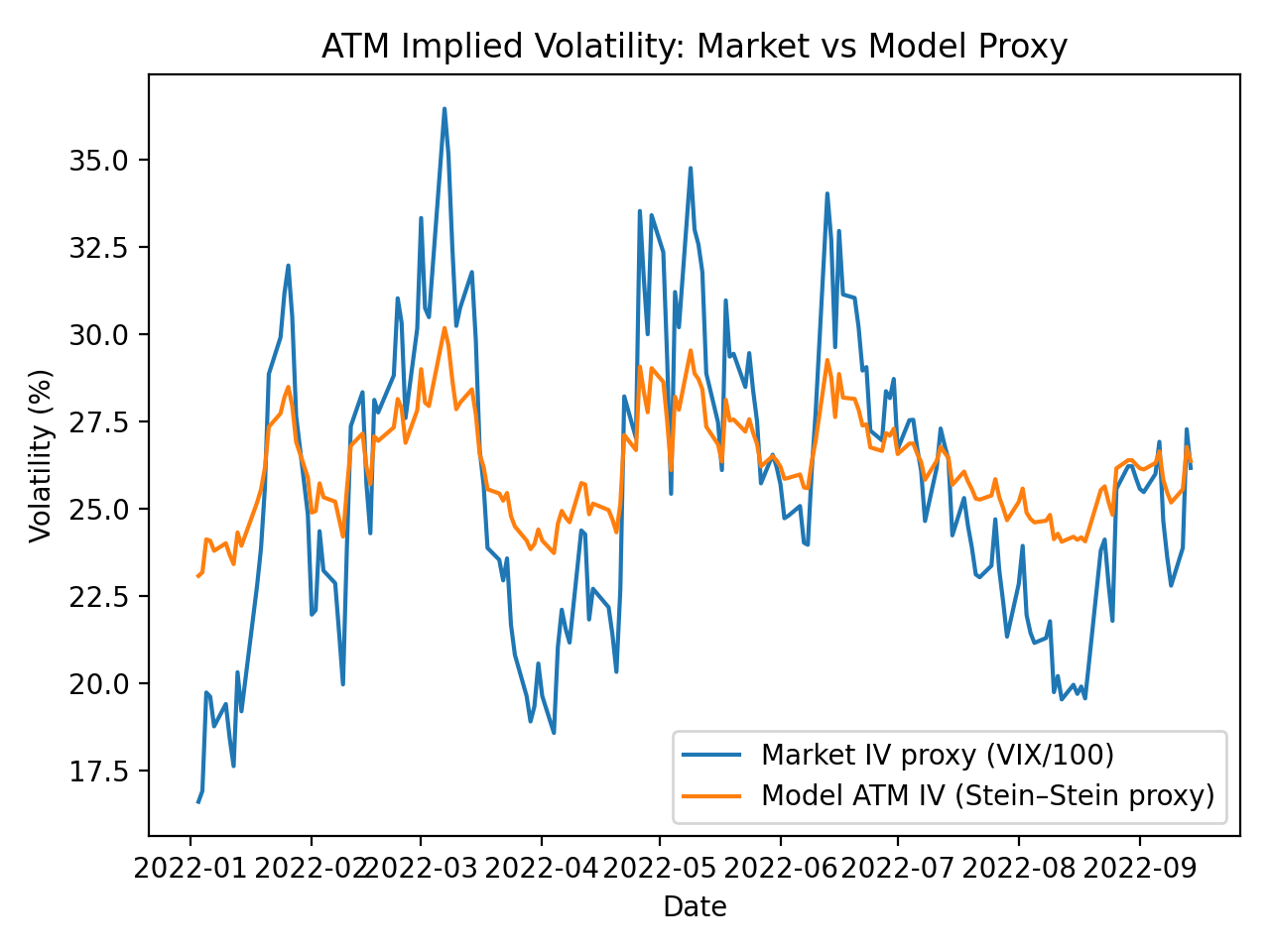}
	\caption{ATM implied volatility proxy. Market proxy is $\sigma_t^{\mathrm{mkt}}=\mathrm{VIX}_t/100$.
		Model proxy $\widehat I_t(\tau)$ uses the Stein-Stein (OU) conditional expected RMS volatility with $\tau=30/365$.}
	\label{fig:iv_market_vs_model}
\end{minipage}
\end{figure}

Figure 1 compares the market-implied volatility proxy,
$\sigma_t^{\mathrm{mkt}}=\mathrm{VIX}_t/100$,
with the model-implied at-the-money volatility generated by the Stein-Stein specification. Several observations can be made from the figure.
First, the model successfully captures the overall level and long-term trend of implied volatility throughout the sample period. In particular, the model reproduces the gradual increase in volatility observed during the first quarter of 2022 and the subsequent stabilization in the later months of the sample.
Second, the model produces a smoother volatility trajectory than the market proxy. While the market-implied volatility exhibits abrupt spikes and short-lived fluctuations, the model-implied volatility evolves more gradually due to the mean-reverting nature of the Ornstein-Uhlenbeck process. This smoothing effect is desirable from a risk-management perspective because it filters out temporary market noise and focuses on the underlying volatility trend.
Third, although the model tracks the average behavior of implied volatility reasonably well, it tends to underestimate extreme market movements. This limitation is particularly noticeable during episodes of sharp increases in the VIX index, where the model responds more slowly than the market. The discrepancy suggests that a pure diffusion volatility model may not be sufficient to explain sudden changes in market sentiment.

This observation motivates the incorporation of a L\'evy jump component in the proposed framework. By allowing simultaneous jumps in both asset returns and volatility, the model can potentially capture abrupt changes in market conditions more accurately than traditional diffusion-based stochastic volatility models. Consequently, the proposed L\'evy-Malliavin framework provides a richer and more realistic description of implied volatility dynamics than the classical Stein-Stein model while retaining analytical tractability.

Overall, Table 2 and Figure 1 demonstrate that the proposed framework is capable of reproducing the broad characteristics of market-implied volatility, including mean reversion and volatility clustering. At the same time, they highlight the importance of incorporating jump effects to better account for extreme market movements and volatility spikes observed in practice.

\section{Conclusion}
The theoretical developments presented in this paper contribute significantly to the advancement of option pricing within the Black-Scholes Model. By proposing a framework that incorporates market observations and implied volatility, this research marks a crucial step toward improving a model that has long remained stagnant. 
The classical formula for option pricing has made progress in the error between both modeled and observed market volatility in the BSM framework. By incorporating the jump term, the model captured the fat-tailed distribution of stock returns and the clustered volatility observed in real markets, resulting in an accurate representation of implied volatility.

Notably, through the utilization of Malliavin calculus, an exact expression for the at-the-money implied volatility of the European call option has been derived under the proposed model. This precise expression considers the dynamics of stock price and volatility, enabling the model to better capture the observed correlation between stock price jumps and volatility jumps in real-world market data.

\end{document}